\documentclass[preprint,11pt]{elsarticle}

\usepackage{amsmath,amssymb,amsthm}    
\usepackage{mathtools}
\usepackage{bm} 

\usepackage{fancyhdr} 
\usepackage{xcolor} 
\usepackage{graphicx} 
\usepackage{float}
\usepackage{subfigure}
\usepackage{hyperref} 
\usepackage{tcolorbox} 
\usepackage[capitalize]{cleveref}

\usepackage{enumitem}          
\usepackage{listings} 

\newtheorem{prop}{Proposition}

\newtheorem{rmk}{Remark}

\begin{document}

\title{A Quantum Kinetic Monte Carlo Method For Lindblad Equation}

\author{Limin Xu}
\ead{xlm20@mails.tsinghua.edu.cn}



\address{Department of Mathematical Sciences, Tsinghua University, Beijing 100084, China}


\begin{abstract}
    In this paper, we generalize the Quantum Kinetic Monte Carlo (QKMC) method of the Schrodinger equation, which was first proposed by [Z.~Cai and J.~Lu. \newblock {\em SIAM J. Sci. Comput.}, 40(3): B706--B722, 2018] to the Lindblad equation. This algorithm makes full use of the tensor product structure of the matrices in the Lindblad equation, thus significantly reducing the storage cost, and can calculate a more extensive system than the existing methods. We demonstrate the method in the framework of the dissipative Ising model, and numerical experiments verify the method's validity and error analysis.
\end{abstract}


\maketitle

\section{Introduction}
For a closed quantum system, the system's state can be described by a wave vector $|\psi\rangle$ in a Hilbert space which satisfies the Schr$\mathrm{\ddot o}$dinger equation. Equivalently, it also can be described by a matrix density $\rho=|\psi\rangle\langle\psi|$, which satisfies the following Von-Neumann equation \cite{breuer_theory_2002,rivas_open_2012}
\begin{equation}
    \partial_t\rho(t)=-\mathrm{i}[H,\rho(t)].
\end{equation}
However, there is no absolute closed quantum system in reality. Every quantum system (which is referred to as "system" in this context) under study interacts with its environment (which is referred to as "environment" in this context), and the system and environment as a whole make up a closed quantum system (which is referred to as "total system" in this context). Such a quantum system is called an open quantum system, which plays a vital role in quantum optics, quantum decoherence, and chemical physics\cite{breuer_colloquium_2016,de_vega_dynamics_2017,johansson_qutip_2012,kramer_quantumopticsjl_2018,noauthor_exploring_nodate,rotter_review_2015,weimer_simulation_2021}. 

For an open quantum system, the key quantity is the density matrix of the system $\rho_{\mathrm{S}}$, which is the partial trace of $\rho_{\mathrm{T}}$:
\begin{equation}
    \rho_{\mathrm{S}}=\operatorname{Tr}_{\mathrm{E}}(\rho_{\mathrm{T}}),
\end{equation}
where $\rho_{\mathrm{T}}$ is the density matrix of the total system, which is a closed quantum system, thus satisfying the Von-Neumann equation:
\begin{equation}
    \partial_t\rho_{\mathrm{T}}(t)=-\mathrm{i}[H_{\mathrm{T}},\rho_{\mathrm{T}}(t)],
\end{equation}
where the $H_{\mathrm{T}}$ is the Hamiltonian of the total system, which can be expressed as
\begin{equation}
    H_{\mathrm{T}}=H_{\mathrm{S}}+H_{\mathrm{E}}+\lambda H_{\mathrm{I}},
\end{equation}
in which $H_{\mathrm{S}}$ is the the Hamiltonian of system, $H_{\mathrm{E}}$ is the the Hamiltonian of environment and $H_{\mathrm{I}}$ is the interaction between system and environment, $\lambda$ is the strength of the interaction. Under the assumption of weak coupling, the $\rho_{\mathrm{S}}$ satisfies the Lindblad equation \cite{manzano_short_2020}:
\begin{equation}
\partial_t \rho_S=\mathcal{L}\left(\rho_S\right):=-\mathrm{i}\left[H_S, \rho_S\right]+\sum_{i=1}^{N}\left(c_i \rho_S c_i^{\dagger}-\frac{1}{2}\left\{c_i^{\dagger} c_i, \rho_S\right\}\right),
\end{equation}
where the $c_i$ is called the jump operator or dissipative operator, which models the interaction with the environment. It is well known that the generator of a completely positive (CP) trace-preserving semigroup
dynamics must have the Lindbladian form.

The Lindblad equation is an ordinary differential equation that can be solved by deterministic solvers such as the Runge-Kuta method, exponential method, and so on. For the error analysis and absolute stability analysis of deterministic solvers, we refer to \cite{cao_structure-preserving_2021}. However, the size of the matrices in the equation grows exponentially with the system, which leads to the disaster of dimensionality in deterministic methods\cite{weimer_simulation_2021,minganti_arnoldi-lindblad_2022}. In practice, when the number of sites is more than $20$, the storage costs in deterministic solvers are beyond the capacity of personal computers. In order to reduce storage costs, there are two main methods were developed\cite{cao_stochastic_2018}. The first method was the Monte Carlo method, which is based on stochastic Schr$\mathrm{\ddot o}$dinger equation (SSE)\cite{cao_stochastic_2018,weimer_simulation_2021}. This method unravels the density matrix to the ensemble of some wave vectors and calculates the density matrix by solving SSE. The second method is the dynamical low-rank method which approximates the density matrix by a lower-rank matrix and calculates the density matrix by solving the equations satisfied by two smaller matrices\cite{bris_low_2013,le_bris_adaptive_2015,le_bris_low-rank_2013}. Both methods reduce the storage cost of the density matrix but still need to store the matrices $H_{\mathrm{S}}$ and $c_i$. Combining the two methods is also considered in \cite{cao_stochastic_2018,le_bris_adaptive_2015}. In a word, although the Monte Carlo method based on SSE and the dynamical low-rank method can reduce some storage costs, existing methods require storing the matrices $H_{\mathrm{S}}$ and $c_i$ and cannot calculate a more extensive system.

In this article, we propose a new Monte Carlo method called the quantum kinetic Monte Carlo method (QKMC) to reduce the computational costs of the Lindblad equation. This method generalizes the quantum kinetic Monte Carlo (QKMC) method, which was first proposed in Schr$\mathrm{\ddot o}$dinger equation by \cite{cai_quantum_2018} to Lindblad equation and makes full use of the tensor product structure of the matrices in the Lindblad equation, thus significantly reducing the storage cost, and can calculate a more extensive system than the existing methods.

The rest of the paper is organized as follows. In \cref{sec:QKMC}, we shall first review the QKMC method in Schr$\mathrm{\ddot o}$dinger equation and then formulate the QKMC method in the Lindblad equation. Then in \cref{sec:model}, we give the dissipative Ising model and the choice of stochastic part in the QKMC method. The error analysis will be provided in \cref{sec:error}. Then numerical results will be presented in \cref{sec:numerical results} to demonstrate the performance of the QKMC method. In \cref{sec:conclusion}, we will give a brief summary and some future work.

\section{Quantum kinetic Monte Carlo Algorithm}\label{sec:QKMC}
\subsection{Quantum kinetic Monte Carlo algorithm for Many-Body Schr$\mathrm{\ddot o}$dinger Equation}

Before We turn to the quantum kinetic Monte Carlo method for Lindblad Equations, let us review the framework of the quantum kinetic Monte Carlo algorithm for many-body Schr$\mathrm{\ddot o}$dinger equation. For more details, we refer to \cite{cai_quantum_2018}.

Given a Hamiltonian $H$ of a closed quantum system, the many-body wave function $|\psi(t)\rangle$ in the Hilbert space $\mathcal{H}$ is governed by the Schr$\mathrm{\ddot o}$dinger equation 
\begin{equation}
    \frac{\mathrm{d}}{\mathrm{d}t}|\psi(t)\rangle=-\mathrm{i}H|\psi(t)\rangle.
\end{equation}
The quantum kinetic Monte Carlo(QKMC) algorithm starts with a decomposition of the Hamiltonian
\begin{equation}
    H=H_{\mathrm{e}}+H_{\mathrm{h}}
\end{equation}
with a class of states $\mathcal{H}_1\subset\mathcal{H}$. This decomposition and $\mathcal{H}_1$ need to satisfy the following conditions:
\begin{itemize}
    \item[(1)] Any vector $|\psi\rangle\in\mathcal{H}_1$ is easy to represent so that we don't have to spend a lot of space to store it, and $\mathcal{H}_1$ might not be a vector space.
    \item[(2)] For any $|\psi\rangle\in\mathcal{H}_1$, the action of $H_{\mathrm{e}}$ remains in $\mathcal{H}_1$,
    \begin{equation}
        \mathrm{e}^{-\mathrm{i}t H_{\mathrm{e}}}|\psi\rangle\in\mathcal{H}_1\qquad \forall t,
    \end{equation}
    and is easy to obtain.
    \item[(3)] There exists a stochastic operator $S(\omega)$ with $\omega$ corresponding to some random space $\Omega$, such that for any $|\psi\rangle\in\mathcal{H}_1$,
    \begin{equation}
        \mathbb{E}_\omega S(\omega)|\psi\rangle=H_{\mathrm {h }}|\psi\rangle \quad \text { and } \quad S(\omega)|\psi\rangle \in \mathcal{H}_1 \quad \forall \omega .
    \end{equation}
\end{itemize}

After such a decomposition, the Schrodinger equation can be written as the following integral form by Duhamel's principle:
\begin{equation}\label{eq:integral formula of wave function}
\begin{aligned}
|\psi(t)\rangle= & \sum_{M=0}^{+\infty} \int_0^t \int_{\Omega} \int_0^{t_M} \int_{\Omega} \cdots \int_0^{t_2} \int_{\Omega}(-\mathrm{i})^M \\
& \times U\left(t, t_M\right) S\left(t_M, \omega_M\right) U\left(t_M, t_{M-1}\right) S\left(t_{M-1}, \omega_{M-1}\right) \cdots \\
& \times U\left(t_2, t_1\right) S\left(t_1, \omega_1\right) U\left(t_1, 0\right)|\psi(0)\rangle \mathrm{d} \mu_{\omega_1} \mathrm{~d} t_1 \cdots \mathrm{d} \mu_{\omega_{M-1}} \mathrm{~d} t_{M-1} \mathrm{~d} \mu_{\omega_M} \mathrm{~d} t_M,
\end{aligned}
\end{equation}
where we used $\mu_{\omega}$ to denote the probability measure of $\omega$ and write
\begin{equation}
\mathbb{E}_\omega f(\omega)=\int_{\Omega} f(\omega) \mathrm{d} \mu_\omega ,
\end{equation}
and $U(t,s)$ is the unitary evolution operator which corresponds to the dynamics governed by $H_{\mathrm{e}}$:
\begin{equation}
    \frac{\mathrm{d}}{\mathrm{d}t}|\Psi(t)\rangle=-\mathrm{i}H_{\mathrm{e}}|\Psi(t)\rangle.
\end{equation}

The expansion \cref{eq:integral formula of wave function} allows us to link the integral to a marked point process with mark space $\Omega$  and then use the Monte Carlo method to evaluate $|\psi(t)\rangle$. We denote $\Xi=\left(\left(t_m\right),\left(\omega_m\right)\right)_{m \geqslant 1}$ one realization of the marked point process. The marked point process is generated by an intensity function $\lambda(t,\omega)$, i.e.,
\begin{equation}
\mathbb{P}(\text { A mark in } \Sigma \text { appears in }[t, t+h))=\int_{\Sigma} \lambda(t, \omega) h \mathrm{~d} \mu_\omega+o(h) \quad \forall t>0, \quad \forall \Sigma \subset \Omega \text {. }
\end{equation}
With this marked point process, we have the following identity for any function $F(\Xi)$ \cite{jacobsen_point_2006}:
\begin{equation}
\begin{aligned}
\mathbb{E}_{\Xi} F(\Xi)= & \sum_{M=0}^{+\infty} \int_0^t \int_{\Omega} \int_0^{t_M} \int_{\Omega} \cdots \int_0^{t_2} \int_{\Omega} \exp \left(-\int_0^t \int_{\Omega} \lambda(s, \omega) \mathrm{d} \mu_\omega \mathrm{d} s\right) \\
& \times\left(\prod_{m=1}^M \lambda\left(t_m, \omega_m\right)\right) F(\Xi) \mathrm{d} \mu_{\omega_1} \mathrm{~d} t_1 \cdots \mathrm{d} \mu_{\omega_{M-1}} \mathrm{~d} t_{M-1} \mathrm{~d} \mu_{\omega_M} \mathrm{~d} t_M .
\end{aligned}
\end{equation}
Therefore, the wave function can be written as
\begin{equation}\label{eq:expectation formula}
|\psi(t)\rangle=\mathbb{E}_{\Xi}\left|\phi_{\Xi}(t)\right\rangle,
\end{equation}
where
\begin{equation}
\begin{aligned}
\left|\phi_{\Xi}(t)\right\rangle= & \exp \left(\int_0^t \int_{\Omega} \lambda(s, \omega) \mathrm{d} \mu_\omega \mathrm{d} s\right) U\left(t, t_M\right) \widetilde{S}\left(t_M, \omega_M\right) \\
& \times U\left(t_M, t_{M-1}\right) \widetilde{S}\left(t_{M-1}, \omega_{M-1}\right) \cdots U\left(t_2, t_1\right) \widetilde{S}\left(t_1, \omega_1\right) U\left(t_1, 0\right)|\psi(0)\rangle,
\end{aligned}
\end{equation}
where $M$ is the number of marks in $\Xi$ before time $t$, and we have used the short-hand
\begin{equation}
\widetilde{S}\left(t_m, \omega_m\right)=-\mathrm{i} S\left(t_m, \omega_m\right) / \lambda\left(t_m, \omega_m\right), \quad m=1, \ldots, M.
\end{equation}

With the \cref{eq:expectation formula}, we can design an Monte Carlo method to evaluate $|\psi(t)\rangle$. The algorithm is to generate a sequence of realizations of process $\Xi$ and evaluate $|\phi_{\Xi}(t)\rangle$ for each realization; an estimate of $|\psi\rangle$ is then given by average. For easier implementation, we define
\begin{equation}
\eta(t)=\int_0^t \int_{\Omega} \lambda(s, \omega) \mathrm{d} \mu_\omega \mathrm{d} s, \quad\left|\tilde{\phi}_{\Xi}(t)\right\rangle=\mathrm{e}^{-\eta(t)}\left|\phi_{\Xi}(t)\right\rangle ,
\end{equation}
then $\eta$ and $|\tilde{\phi}_{\Xi}(t)$ satisfies
\begin{equation}\label{eq:eta}
\frac{\mathrm{d} \eta}{\mathrm{d} t}=\int_{\Omega} \lambda(t, \omega) \mathrm{d} \mu_\omega,
\end{equation}
and
\begin{equation}\label{eq:Phi_tilde}
    \frac{\mathrm{d}}{\mathrm{d} t}\left|\tilde{\phi}_{\Xi}(t)\right\rangle=-\mathrm{i} H_{\text {e}}(t)\left|\tilde{\phi}_{\Xi}(t)\right\rangle, \quad t \in\left(t_{m-1}, t_m\right), \quad m>0.
\end{equation}
Then we have
\begin{equation}
\mathbb{P}\left(\text { no mark exists in }\left(s_1, s_2\right)\right)=\exp \left(-\left[\eta\left(s_2\right)-\eta\left(s_1\right)\right]\right) .
\end{equation}

Finally, we give the QKMC algorithm for the many-body Schrodinger equation based on the above notation:
\begin{itemize}
    \item[(1)] Choose time step $\Delta t$, set $t\gets 0, \eta(t)\gets 0, |\tilde{\phi}_{\Xi}(t)\rangle\gets |\psi(0)\rangle$. 
    \item[(2)] Solve $\left|\tilde{\phi}_{\Xi}(t+\Delta t)\right\rangle$ according to \cref{eq:Phi_tilde}.
    \item[(3)] If $t+\Delta t=T$, stop. Otherwise, generate a random number $Y$ obeying the uniform distribution in $[0,1]$ and solve $\eta(t+\Delta t)$ according to \cref{eq:eta}, then 
    \begin{itemize}
        \item[(3.1)] If $\exp{(\eta(t)-\eta(t+\Delta t))}\le 1-Y$, set $t\gets t+\Delta t, \eta(t)\gets \eta(t+\Delta t), |\tilde{\phi}_{\Xi}(t)\rangle\gets |\tilde{\phi}_{\Xi}(t+\Delta t)\rangle$ and return to step $(2)$.
        \item[(3.2)] Otherwise, generate a mark $\omega\in\Omega$ according to the probability measure $\mu_{\omega}$, set $|\widetilde{\phi}_{\Xi}(t+\Delta t)\rangle \leftarrow \widetilde{S}(t+\Delta t, \omega)|\widetilde{\phi}_{\Xi}(t+\Delta t)\rangle, t\gets t+\Delta t, \eta(t)\gets \eta(t+\Delta t), |\tilde{\phi}_{\Xi}(t)\rangle\gets |\tilde{\phi}_{\Xi}(t+\Delta t)\rangle$ and return to step $(2)$.
    \end{itemize}
\end{itemize}

\subsection{Quantum kinetic Monte Carlo algorithm for Lindblad Equation}
We will generalize the QKMC algorithm to the Lindblad equation in this subsection. We denote the set of all density matrices as $\mathcal{M}$.
Similar to the Schr$\mathrm{\ddot o}$dinger equation, we start with a decomposition of the Liouville operator
\begin{equation}
    \mathcal{L}=\mathcal{L}_{\mathrm{e}}+\mathcal{L}_{\mathrm{h}}
\end{equation}
with a class of density matrix $\mathcal{M}_1\subset\mathcal{M}$. This decomposition and $\mathcal{M}_1$ need to satisfy the following conditions:
\begin{itemize}
    \item[(1)] Any density matrix ${\rho}_S\in\mathcal{M}_1$ is easy to represent so that we don't have to spend a lot of space to store it, and $\mathcal{M}_1$ might not be a vector space.
    \item[(2)] For any ${\rho}_S\in\mathcal{M}_1$, the action of $\mathcal{L}_{\mathrm{e}}$ remains in $\mathcal{M}_1$, 
    \begin{equation}
        U(t,s)({\rho}_S)\in\mathcal{M}_1\quad \forall t,
    \end{equation}
    and is easy to obtain. The evolution operators $U(t,s)$ corresponds to the dynamics governed by $\mathcal{L}_{\mathrm{e}}$:
    \begin{equation}
        \frac{\mathrm{d}}{\mathrm{d}t}{\rho}_S(t)=\mathcal{L}_{\mathrm{e}}({\rho}_S(t)).
    \end{equation}
    \item[(3)] There exists a stochastic super-operator $S(\omega)$ with $\omega$ corresponding to some random space $\Omega$, such that for any ${\rho}_S\in\mathcal{M}_1$,
    \begin{equation}
        \mathbb{E}_\omega S(\omega)({\rho}_S)=\mathcal{L}_{\mathrm {h }}({\rho}_S) \quad \text { and } \quad S(\omega)({\rho}_S) \in \mathcal{M}_1 \quad \forall \omega .
    \end{equation}
\end{itemize}

After such a decomposition, the Lindblad equation can be written as the following integral form by Duhamel's principle:
\begin{equation}\label{eq:integral formula of density matrix}
\begin{aligned}
{\rho}_S(t)= & \sum_{M=0}^{+\infty} \int_0^t \int_{\Omega} \int_0^{t_M} \int_{\Omega} \cdots \int_0^{t_2} \int_{\Omega} \\
& \times U\left(t, t_M\right) S\left(t_M, \omega_M\right) U\left(t_M, t_{M-1}\right) S\left(t_{M-1}, \omega_{M-1}\right) \cdots \\
& \times U\left(t_2, t_1\right) S\left(t_1, \omega_1\right) U\left(t_1, 0\right)({\rho}_S(0)) \mathrm{d} \mu_{\omega_1} \mathrm{~d} t_1 \cdots \mathrm{d} \mu_{\omega_{M-1}} \mathrm{~d} t_{M-1} \mathrm{~d} \mu_{\omega_M} \mathrm{~d} t_M,
\end{aligned}
\end{equation}
where we used $\mu_{\omega}$ to denote the probability measure of $\omega$ and write
\begin{equation}
\mathbb{E}_\omega f(\omega)=\int_{\Omega} f(\omega) \mathrm{d} \mu_\omega .
\end{equation}

The expansion \cref{eq:integral formula of density matrix} allows us to use the Monte Carlo method to evaluate ${\rho}_S(t)$. Like the Schr$\mathrm{\ddot o}$dinger equation, We denote $\Xi=\left(\left(t_m\right),\left(\omega_m\right)\right)_{m \geqslant 1}$ one realization of the marked point process which is generated by an intensity function $\lambda(t,\omega)$.
Then, the density matrix (we drop the subscript $S$ for simplicity ) can be written as
\begin{equation}\label{eq:expectation formula of density matrix}
{\rho}(t)=\mathbb{E}_{\Xi}{\rho}_{\Xi}(t),
\end{equation}
where
\begin{equation}
\begin{aligned}
{\rho}_{\Xi}(t)= & \exp \left(\int_0^t \int_{\Omega} \lambda(s, \omega) \mathrm{d} \mu_\omega \mathrm{d} s\right) U\left(t, t_M\right) \widetilde{S}\left(t_M, \omega_M\right) \\
& \times U\left(t_M, t_{M-1}\right) \widetilde{S}\left(t_{M-1}, \omega_{M-1}\right) \cdots U\left(t_2, t_1\right) \widetilde{S}\left(t_1, \omega_1\right) U\left(t_1, 0\right){\rho}_{\Xi}(0),
\end{aligned}
\end{equation}
where $M$ is the number of marks in $\Xi$ before time $t$, and we have used the short-hand
\begin{equation}
\widetilde{S}\left(t_m, \omega_m\right)(\cdot)= S\left(t_m, \omega_m\right)(\cdot) / \lambda\left(t_m, \omega_m\right), \quad m=1, \ldots, M.
\end{equation}

We define
\begin{equation}
\eta(t)=\int_0^t \int_{\Omega} \lambda(s, \omega) \mathrm{d} \mu_\omega \mathrm{d} s, \quad
\widetilde{{\rho}}_{\Xi}(t)=\mathrm{e}^{-\eta(t)}{\rho}_{\Xi}(t) ,
\end{equation}
then $\eta$ and $|\tilde{\rho}_{\Xi}(t)$ satisfies
\begin{equation}\label{eq:eta of density matrix}
\frac{\mathrm{d} \eta}{\mathrm{d} t}=\int_{\Omega} \lambda(t, \omega) \mathrm{d} \mu_\omega,
\end{equation}
and
\begin{equation}\label{eq:rho_tilde}
    \frac{\mathrm{d}}{\mathrm{d} t}\tilde{\rho}_{\Xi}(t)=\mathcal{L}_{\mathrm{e}}(\tilde{\rho}_{\Xi}(t)), \quad t \in\left(t_{m-1}, t_m\right), \quad m>0.
\end{equation}
Then we have
\begin{equation}
\mathbb{P}\left(\text { no mark exists in }\left(s_1, s_2\right)\right)=\exp \left(-\left[\eta\left(s_2\right)-\eta\left(s_1\right)\right]\right) .
\end{equation}

Finally, with the \cref{eq:expectation formula of density matrix}, we can design the QKMC algorithm to evaluate ${\rho}(t)$ based on the above notation:
\begin{itemize}
    \item[(1)] Choose time step $\Delta t$, set $t\gets 0, \eta(t)\gets 0, |\tilde{\rho}_{\Xi}(t)\rangle\gets |\rho(0)\rangle$. 
    \item[(2)] Solve $\left|\tilde{\rho}_{\Xi}(t+\Delta t)\right\rangle$ according to \cref{eq:rho_tilde}.
    \item[(3)] If $t+\Delta t=T$, stop. Otherwise, generate a random number $Y$ obeying the uniform distribution in $[0,1]$ and solve $\eta(t+\Delta t)$ according to \cref{eq:eta of density matrix}, then 
    \begin{itemize}
        \item[(3.1)] If $\exp{(\eta(t)-\eta(t+\Delta t))}\le 1-Y$, set $t\gets t+\Delta t, \eta(t)\gets \eta(t+\Delta t), |\tilde{\rho}_{\Xi}(t)\rangle\gets |\tilde{\rho}_{\Xi}(t+\Delta t)\rangle$ and return to step $(2)$.
        \item[(3.2)] Otherwise, generate a mark $\omega\in\Omega$ according to the probability measure $\mu_{\omega}$, set $|\widetilde{\rho}_{\Xi}(t+\Delta t)\rangle \leftarrow \widetilde{S}(t+\Delta t, \omega)|\widetilde{\rho}_{\Xi}(t+\Delta t)\rangle, t\gets t+\Delta t, \eta(t)\gets \eta(t+\Delta t), |\tilde{\rho}_{\Xi}(t)\rangle\gets |\tilde{\rho}_{\Xi}(t+\Delta t)\rangle$ and return to step $(2)$.
    \end{itemize}
\end{itemize}

\section{The dissipative Ising model}\label{sec:model}
In this section, we give the dissipative Ising model and present the details of the QKMC method for the Lindblad equation based on this model.

The Hamiltonian of the dissipative Ising model satisfies\cite{weimer_simulation_2021} 
\begin{equation}\label{eq:Dissipative Ising model}
\begin{aligned}
    {H}&=\frac{g}{2} \sum_{i=1}^N \sigma_x^{(i)}+\frac{h}{2} \sum_{i=1}^N \sigma_z^{(i)}+\frac{V}{4} \sum_{\langle i j\rangle} \sigma_z^{(i)} \sigma_z^{(j)},
\end{aligned}
\end{equation}
where $\langle ij\rangle$ means that sites $i$ and $j$ are adjacent, $\gamma$ is the dissipative rate, and
\begin{equation}
\begin{aligned}
    \sigma_x^{(i)}&=\mathrm{Id}^{\otimes(i-1)} \otimes \sigma_x \otimes \mathrm{Id}^{\otimes(N-i)},\\
    \sigma_z^{(i)}&=\mathrm{Id}^{\otimes(i-1)} \otimes \sigma_z \otimes \mathrm{Id}^{\otimes(N-i)},\\
    \sigma_z^{(i)} \sigma_z^{(j)}&=\mathrm{Id}^{\otimes(i-1)} \otimes \sigma_z \otimes \mathrm{Id}^{\otimes(j-i-1)} \otimes \sigma_z \otimes \mathrm{Id}^{\otimes(N-j)},
\end{aligned}
\end{equation}
where $\mathrm{Id}$ is the $2\times2$ identity matrix, and
\begin{equation}
    \sigma_{x}=
    \begin{bmatrix}
        0&1\\
        1&0
    \end{bmatrix},
    \sigma_{z}=
    \begin{bmatrix}
        1&0\\
        0&-1
    \end{bmatrix}.
\end{equation}
The jump operators satisfy
\begin{equation}
    {c}_i=\mathrm{Id}^{\otimes(i-1)} \otimes \sqrt{\gamma}\sigma_{-} \otimes \mathrm{Id}^{\otimes(N-i)},
\end{equation}
where
\begin{equation}
    \sigma_{-}=
    \begin{bmatrix}
        0&0\\
        1&0
    \end{bmatrix}.
\end{equation}

For this model, We denote $\mathcal{L}_1,\mathcal{L}_2,\mathcal{L}_3,\mathcal{L}_4,\mathcal{L}_5$ as follows
\begin{equation}
\begin{aligned}
    &\mathcal{L}_1({\rho})=-\mathrm{i}\left[\frac{g}{2} \sum_i \sigma_x^{(i)},{\rho}\right],\quad 
    \mathcal{L}_2({\rho})=-\mathrm{i}\left[\frac{h}{2} \sum_i \sigma_z^{(i)},{\rho}\right],\quad 
    \mathcal{L}_3({\rho})=-\mathrm{i}\left[\frac{V}{4} \sum_{\langle ij\rangle} \sigma_z^{(i)}\sigma_z^{(j)},{\rho}\right],\\
    &\mathcal{L}_4({\rho})=\sum_i{c}_i {\rho} {c}_i^{\dagger},\quad
    \mathcal{L}_5({\rho})=-\frac{1}{2}\sum_i\left\{{c}_i^{\dagger} {c}_i, {\rho}\right\},
\end{aligned}
\end{equation}
then we have
\begin{equation}
    \mathcal{L}(\hat{\rho})=\mathcal{L}_1(\hat{\rho})+\mathcal{L}_2(\hat{\rho})+\mathcal{L}_3(\hat{\rho})+\mathcal{L}_4(\hat{\rho})+\mathcal{L}_5(\hat{\rho}).
\end{equation}

We choose the class $\mathcal{M}_1$ as follows:
\begin{equation}
\mathcal{M}_1=\left\{{\rho}={\rho}_1 \otimes \cdots \otimes{\rho}_N|{\rho}_i\quad\text{is a $2\times2$ matrix}, \forall i\right\}.
\end{equation}
It is evident that any density matrix ${\rho}\in\mathcal{M}$ can be written as the linear combination of elements in $\mathcal{M}_1$, but it's worth noting that neither $\mathcal{M}_1$ nor $\mathcal{M}$ are vector Spaces. The following proposition is essential for the design of the QKMC algorithm.
\begin{prop}\label{prop: tensor prop}
    If $U(t,s)$ is the evolution operator, which corresponds to the following dynamics:
    \begin{equation}
        \partial_t\rho=\mathcal{L}(\rho),
    \end{equation}
    where $\mathcal{L}(\rho)$ has the following expression when $\rho=\rho_1\otimes\rho_2\otimes\cdots\otimes\rho_N$:
    \begin{equation}
    \begin{aligned}
        \mathcal{L}(\rho_1\otimes\rho_2\otimes\cdots\otimes\rho_N)
        &=(\mathcal{L}^{(1)}(\rho_1)\otimes\rho_2\otimes\cdots\otimes\rho_N)+(\rho_1\otimes\mathcal{L}^{(2)}(\rho_2)\otimes\cdots\otimes\rho_N)\\
        &+\cdots+(\rho_1\otimes\rho_2\otimes\cdots\otimes\mathcal{L}^{(N)}(\rho_N)),
    \end{aligned}
    \end{equation}
    and $U^{(k)}(t,s),k=1,2,\cdots, N$ are the evolution operators which correspond to the following dynamics, respectively:
    \begin{equation}
        \partial_t(\rho_k)=\mathcal{L}^{(k)}(\rho_k),k=1,2,\cdots, N,
    \end{equation}
    then we have
    \begin{equation}\label{eq: tensor prop}
        U(t,s)\left(\rho_1(s)\otimes\rho_2(s)\otimes\cdots\otimes\rho_N(s)\right)=U^{(1)}(t,s)\rho_1(s)\otimes U^{(2)}(t,s)\rho_2(s)\otimes\cdots\otimes U^{(N)}(t,s)\rho_N(s).
    \end{equation}
\end{prop}
\begin{proof}
    By direct calculation, we have
    \begin{equation}
    \begin{aligned}
         &\frac{\mathrm{d}}{\mathrm{d}t}\left(U^{(1)}(t,s)\rho_1(s)\otimes U^{(2)}(t,s)\rho_2(s)\otimes\cdots\otimes U^{(N)}(t,s)\rho_N(s)\right)\\
         &=\left(\frac{\mathrm{d}}{\mathrm{d}t}\left(U^{(1)}(t,s)\rho_1(s)\right)\otimes U^{(2)}(t,s)\rho_2(s)\otimes\cdots\otimes U^{(N)}(t,s)\rho_N(s)\right)\\
         &+\left(U^{(1)}(t,s)\rho_1(s)\otimes \frac{\mathrm{d}}{\mathrm{d}t}\left( U^{(2)}(t,s)\rho_2(s)\right)\otimes\cdots\otimes U^{(N)}(t,s)\rho_N(s)\right)\\
         &+\cdots\\
         &+\left(U^{(1)}(t,s)\rho_1(s)\otimes U^{(2)}(t,s)\rho_2(s)\otimes\cdots\otimes \frac{\mathrm{d}}{\mathrm{d}t} \left(U^{(N)}(t,s)\rho_N(s)\right)\right)\\
         &=\left(\mathcal{L}^{(1)}\left(U^{(1)}(t,s)\rho_1(s)\right)\otimes U^{(2)}(t,s)\rho_2(s)\otimes\cdots\otimes U^{(N)}(t,s)\rho_N(s)\right)\\
         &+\left(U^{(1)}(t,s)\rho_1(s)\otimes \mathcal{L}^{(2)}\left( U^{(2)}(t,s)\rho_2(s)\right)\otimes\cdots\otimes U^{(N)}(t,s)\rho_N(s)\right)\\
         &+\cdots\\
         &+\left(U^{(1)}(t,s)\rho_1(s)\otimes U^{(2)}(t,s)\rho_2(s)\otimes\cdots\otimes \mathcal{L}^{(N)} \left(U^{(N)}(t,s)\rho_N(s)\right)\right)\\
         &=\mathcal{L}\left(U^{(1)}(t,s)\rho_1(s)\otimes U^{(2)}(t,s)\rho_2(s)\otimes\cdots\otimes U^{(N)}(t,s)\rho_N(s)\right).
    \end{aligned}
    \end{equation}
    Therefore, the left and right sides of the \cref{eq: tensor prop}  satisfy the same dynamics equation. When $t=s$, the left side of the \cref{eq: tensor prop} equals the right side, so they have the same initial value. It follows that \cref{eq: tensor prop} valid from the theory of ordinary differential equation.
\end{proof}

According to the above proposition, it follows that the evolution corresponding to $\mathcal{L}_1+\mathcal{L}_2+\mathcal{L}_4+\mathcal{L}_5$ remains the density matrix in $\mathcal{M}_1$, but which corresponding to $\mathcal{L}_3$ change the density matrix out $\mathcal{M}_1$. Therefore, we choose
\begin{equation}
    \mathcal{L}_{\mathrm{e}}=\mathcal{L}_{1}+\mathcal{L}_{2}+\mathcal{L}_{4}+\mathcal{L}_{5},\quad \mathcal{L}_{\mathrm{h}}=\mathcal{L}_{3},
\end{equation}
then the $\mathcal{L}^{(k)}$ satisfies
\begin{equation}
    \mathcal{L}^{(k)}(\rho_k)=-\mathrm{i}\left[\frac{g}{2}\sigma_x+\frac{h}{2}\sigma_z,\rho_k\right]+\gamma\sigma_{-}\rho_k\sigma_{-}^{\dagger}-\frac{\gamma}{2}\left\{\sigma_{-}^{\dagger}\sigma_{-},\rho_k\right\}.
\end{equation}

The sample space can be chosen as 
    \begin{equation}
        \begin{aligned}
            \Omega=\{(j, k,l) \mid j, k=1, \ldots, N, \text{$j<k$ and $j$ is adjancent to $k$},l=1,2\},
        \end{aligned}
    \end{equation}
    and the probability measure is given by
    \begin{equation}
        \mathbb{P}\left(\omega=\omega_0\right)=1 /|\Omega| \quad \forall \omega_0 \in \Omega .
    \end{equation}
    The operator $S(\omega)$ is set to be
    \begin{equation}
        S(t, \omega)({\rho})= \begin{cases}
        -\mathrm{i}|\Omega|\frac{V}{4} \sigma_z^{(j)}\sigma_z^{(k)}{\rho} & \text { if } \omega=(j, k,1) ,\\
        \mathrm{i}|\Omega|\frac{V}{4} {\rho}\sigma_z^{(j)}\sigma_z^{(k)} & \text { if } \omega=(j, k,2) ,
        \end{cases}
    \end{equation}
    where $|\Omega|$ is the cardinality of $\Omega$. The intensity function is chosen as
    \begin{equation}\label{eq:intensity function}
        \lambda(t, \omega)= 
        \begin{cases} 
        |\Omega|\frac{V}{4} & \text { if } \omega=(j, k,1) \in \Omega_2,\\
        |\Omega|\frac{V}{4} & \text { if } \omega=(j, k,2) \in \Omega_2,
        \end{cases}
    \end{equation}
    so that the operator $\widetilde{S}(t,\omega)$  does not change the magnitude of the density matrix.

\section{Analysis of the sampling variance}\label{sec:error}
In this section, we will analyze the evolution of the sampling variance of the QKMC algorithm for the Lindblad equation.  Firstly, there are many norms for matrix, and we should choose one of them to analyze the variance. In this article, we choose the Frobeninus norm:
\begin{equation}
    \|\mathbf{A}\|_F=\sqrt{\operatorname{Tr}\left(\mathbf{A}^{\dagger} \mathbf{A}\right)}=\sqrt{\sum_{i=1}^m \sum_{j=1}^n |A_{i j}|^2}, \quad \mathbf{A} \in \mathbb{C}^{m \times n}.
\end{equation}
We introduce the matrix straightening operation: if $A_{n\times p}=(\vec{a}_1,\vec{a}_2,\cdots,\vec{a}_p)$(where $\vec{a}_i,1\le i\le p$ is $n\times1$ column vector) is a matrix, then $\operatorname{vec}(A)$ is the straightening of matrix $A$ which is defined as follows
\begin{equation}
\operatorname{vec}(A)=\left(\begin{array}{c}
\vec{a}_1 \\
\vec{a}_2 \\
\vdots \\
\vec{a}_p
\end{array}\right).
\end{equation}
Then the Frobeninus norm of matrix $A$ is equivalent to the $2$ norm of its straightening $\operatorname{vec}(A)$, i.e.,
\begin{equation}
    \|\mathbf{A}\|_F=\|\operatorname{vec}(A)\|_2.
\end{equation}
The following proposition is important for our analysis (their proof can be found in \cite{golub2013matrix}):
\begin{prop}\label{prop: straightening}
    The straightening of a matrix has the following properties: 
    \begin{equation}
        \operatorname{vec}(A X B)=\left(B^T \otimes A\right) \operatorname{vec}(X).
    \end{equation}
\end{prop}
\begin{prop}\label{prop:kronecker product}
    The Frobeninus norm of the Kronecker product satisfies the following equation:
    \begin{equation}
        \|A\otimes B\|_{F}=\|A\|_{F}\cdot\|B\|_{F}.
    \end{equation}
\end{prop}
\begin{prop}\label{prop: straightening eigenvalue}
    If $\lambda_i,1\le i\le m,\mu_j,1\le j\le n$ are the corresponding eigenvalues of $A\in\mathbb{C}^{m\times m},B\in\mathbb{C}^{n\times n}$, then $\lambda_i\mu_j,1\le i\le m,1\le j\le n$ are the eigenvalues of $A\otimes B.$
\end{prop}
In the following, for convenience, we will denote $\operatorname{vec}(\rho(t))$ by $|\rho(t)\rangle$ and $\operatorname{vec}(U\left(t, s\right)),\operatorname{vec}(S\left(t, \omega\right))$ by $\hat{U}\left(t, s\right),\hat{S}\left(t, \omega\right)$, where $\operatorname{vec}(U\left(t, s\right)),\operatorname{vec}(S\left(t, \omega\right))$ are defined as follows
\begin{equation}
    \operatorname{vec}(U\left(t, s\right)\rho(t))=\operatorname{vec}(U\left(t, s\right))\operatorname{vec}(\rho(t)),\quad \operatorname{vec}(S\left(t, \omega\right)\rho(t))=\operatorname{vec}(S\left(t, \omega\right))\operatorname{vec}(\rho(t)).
\end{equation}

Using the above notation, we can rewrite \cref{eq:integral formula of density matrix} as follows
\begin{equation}\label{eq:integral formula of straightening density matrix}
\begin{aligned}
|\rho(t)\rangle= & \sum_{M=0}^{+\infty} \int_0^t \int_{\Omega} \int_0^{t_M} \int_{\Omega} \cdots \int_0^{t_2} \int_{\Omega} \\
& \times \hat{U}\left(t, t_M\right) 
\hat{S}\left(t_M, \omega_M\right)
\hat{U}\left(t_M, t_{M-1}\right)
\hat{S}\left(t_{M-1}, \omega_{M-1}\right)
\cdots \\
& \times 
\hat{U}\left(t_2, t_1\right)
\hat{S}\left(t_1, \omega_1\right)
\hat{U}\left(t_1, 0\right)
|\rho(0)\rangle
\mathrm{d} \mu_{\omega_1} \mathrm{~d} t_1 \cdots \mathrm{d} \mu_{\omega_{M-1}} \mathrm{~d} t_{M-1} \mathrm{~d} \mu_{\omega_M} \mathrm{~d} t_M,
\end{aligned}
\end{equation}
and \cref{eq:expectation formula of density matrix} becomes
\begin{equation}\label{eq:expectation formula of straightening density matrix}
|\rho(t)\rangle=\mathbb{E}_{\Xi}|\rho_{\Xi}(t)\rangle,
\end{equation}
where
\begin{equation}\label{eq: single trajectory of rho}
\begin{aligned}
|\rho_{\Xi}(t)\rangle= 
& \exp \left(\int_0^t \int_{\Omega} \lambda(s, \omega) \mathrm{d} \mu_\omega \mathrm{d} s\right) 
\hat{U}\left(t, t_M\right) 
\widehat{\widetilde{S}}\left(t_M, \omega_M\right) \\
& \times \hat{U}\left(t_M, t_{M-1}\right) 
\widehat{\widetilde{S}}\left(t_{M-1}, \omega_{M-1}\right) \cdots \hat{U}\left(t_2, t_1\right) 
\widehat{\widetilde{S}}\left(t_1, \omega_1\right) 
\hat{U}\left(t_1, 0\right)
|\rho_{\Xi}(0)\rangle.
\end{aligned}
\end{equation}

Suppose $N_{\text{traj}}$ trajectories are used in the simulation. Then, the numerical solution is
\begin{equation}
    |\rho_{\text{num}}(t)\rangle=\frac{1}{N_{\text{traj}}}\sum_{i=1}^{N_{\text{traj}}}|\rho_{\Xi_i}(t)\rangle, 
\end{equation}
where $\Xi_i$ is the realization of the marked point process corresponding to the $i$th trajectory, and $\Xi_i$ and $\Xi_j$ are independent if $i\neq j$. Define $|\rho_{\text{err}}(t)\rangle$ as the difference between the numerical solution and the exact solution:
\begin{equation}
    |\rho_{\text{err}}(t)\rangle=|\rho_{\text{num}}(t)\rangle-|\rho(t)\rangle.
\end{equation}
It is obvious that $\mathbb{E}[|\rho_{\text{err}}(t)\rangle]=0$, hence the variance of $|\rho_{\text{err}}(t)\rangle$ is
\begin{equation}\label{eq: variance}
\begin{aligned}
    \sigma^2&=\mathbb{E}[\||\rho_{\text{err}}(t)\rangle\|_2]=\mathbb{E}[\langle\rho_{\text{err}}(t)||\rho_{\text{err}}(t)\rangle]\\
    &=\frac{1}{N_{\text {traj }}^2} \sum_{i=1}^{N_{\text {traj }}} \sum_{j=1}^{N_{\text {traj }}} \mathbb{E}_{\Xi_i, \Xi_j}\left\langle\rho_{\Xi_i}(t) \mid \rho_{\Xi_j}(t)\right\rangle -\langle\rho(t) \mid \rho(t)\rangle \\
    &=\frac{1}{N_{\text {traj }}^2} \sum_{i,j=1,i\neq j}^{N_{\text {traj }}} \mathbb{E}_{\Xi_i, \Xi_j}\left\langle\rho_{\Xi_i}(t) \mid \rho_{\Xi_j}(t)\right\rangle 
    +\frac{1}{N_{\text {traj }}^2} \sum_{i=j=1}^{N_{\text {traj }}} \mathbb{E}_{\Xi_i, \Xi_j}\left\langle\rho_{\Xi_i}(t) \mid \rho_{\Xi_j}(t)\right\rangle -\langle\rho(t) \mid \rho(t)\rangle\\
    &=\frac{N_{\text {traj }}(N_{\text {traj }}-1)}{N_{\text {traj }}^2}\langle\rho(t) \mid \rho(t)\rangle
    +\frac{1}{N_{\text {traj }}} \mathbb{E}_{\Xi}\left\langle\rho_{\Xi}(t) \mid \rho_{\Xi}(t)\right\rangle-\langle\rho(t) \mid \rho(t)\rangle\\
    &=\frac{1}{N_{\text {traj }}} \mathbb{E}_{\Xi}\left\langle\rho_{\Xi}(t) \mid \rho_{\Xi}(t)\right\rangle-\frac{1}{N_{\text {traj }}} \langle\rho(t) \mid \rho(t)\rangle\\
    &\le \frac{1}{N_{\text {traj }}} \mathbb{E}_{\Xi}\left\langle\rho_{\Xi}(t) \mid \rho_{\Xi}(t)\right\rangle.
\end{aligned}
\end{equation}
To estimate $\mathbb{E}_{\Xi}\left\langle\rho_{\Xi}(t) \mid \rho_{\Xi}(t)\right\rangle$, we need the following propositions:
\begin{prop}\label{prop: L^k}
Let $\hat{\mathcal{L}}^{(k)}$ be the straightening of ${\mathcal{L}}^{(k)}$, i.e.,
    \begin{equation}
    \begin{aligned}
        \hat{\mathcal{L}}^{(k)}&=-\mathrm{i}\left(\mathrm{Id}\otimes\left(\frac{g}{2}\sigma_x+\frac{h}{2}\sigma_z\right)-\left(\frac{g}{2}\sigma_x+\frac{h}{2}\sigma_z\right)^{T}\otimes\mathrm{Id}\right)\\
        &+\gamma\sigma_{-}^{*}\otimes\sigma_{-}-\frac{\gamma}{2}\left(\mathrm{Id}\otimes\left(\sigma_{-}^{\dagger}\sigma_{-}\right)+\left(\sigma_{-}^{\dagger}\sigma_{-}\right)^{T}\otimes\mathrm{Id}\right),
    \end{aligned}
    \end{equation}
    then the eigenvalues of $\hat{\mathcal{L}}^{(k)}$ consist of zero and some complex numbers whose real parts are negative, and $\hat{\mathcal{L}}^{(k)}$ is diagonalizable.
\end{prop}
The matrix $\hat{\mathcal{L}}^{(k)}$ is a $4\times4$ matrix so that this proposition can be directly verified manually or using Matlab or Mathematica. Also, this proposition is a direct consequence of \cite[Theorem 4.4.2]{rivas_open_2012} and the fact that the ${\mathcal{L}}^{(k)}$ generates a relaxing completely positive semigroup.

\begin{prop}
    The evolution operator $U(t,s)$ which corresponds to $\mathcal{L}_{\mathrm{e}}$ satisfies
    \begin{equation}
        \|U(t,s)\left(\rho_1(s)\otimes\rho_2(s)\otimes\cdots\otimes\rho_N(s)\right)\|_{F}\le \|\rho_1(s)\otimes\rho_2(s)\otimes\cdots\otimes\rho_N(s)\|_{F}
    \end{equation}
\end{prop}
\begin{proof}
    From \cref{prop: tensor prop} and \cref{prop:kronecker product}, it follows that
    \begin{equation}
        \begin{aligned}
            &\|U(t,s)\left(\rho_1(s)\otimes\rho_2(s)\otimes\cdots\otimes\rho_N(s)\right)\|_{F}\\
            &=\|U^{(1)}(t,s)\rho_1(s)\otimes U^{(2)}(t,s)\rho_2(s)\otimes\cdots\otimes U^{(N)}(t,s)\rho_N(s)\|_{F}\\
            &=\|U^{(1)}(t,s)\rho_1(s)\|_{F}\cdot \|U^{(2)}(t,s)\rho_2(s)\|_{F}\cdots\| U^{(N)}(t,s)\rho_N(s)\|_{F},
        \end{aligned}
    \end{equation}
    and
    \begin{equation}
        \begin{aligned}
            \|\rho_1(s)\otimes\rho_2(s)\otimes\cdots\otimes\rho_N(s)\|_{F}=\|\rho_1(s)\|_{F}\cdot\|\rho_2(s)\|_{F}\cdots\|\rho_N(s)\|_{F}.
        \end{aligned}
    \end{equation}
    Thus, we only need to prove the following inequality:
    \begin{equation}\label{eq:norm less}
        \|U^{(k)}(t,s)\rho_k(s)\|_{F}\le \|\rho_k(s)\|_{F}\iff \|\hat{U}^{(k)}(t,s)|\rho_k(s)\rangle\|_{2}\le \||\rho_k(s)\rangle\|_{2} .
    \end{equation}
From \cref{prop: L^k}, it follows that the real parts of eigenvalues of $\hat{U}^{(k)}(t,s)$ are greater than $0$ and less than $1$, and $\hat{U}^{(k)}(t,s)$ is diagonalizable. Thus, \cref{eq:norm less} holds from the knowledge of linear algebra. This completes the proof.
\end{proof}

According to the above proposition, it follows that,
\begin{equation}\label{eq: inner of straightening rho}
\left\langle\rho_{\Xi}(t) \mid \rho_{\Xi}(t)\right\rangle \leqslant \exp \left(2 \int_0^t \int_{\Omega} \lambda(s, \omega) \mathrm{d} \mu_\omega \mathrm{d} s\right)\left(\prod_{m=1}^M\left\|\widehat{\widetilde{S}}\left(t_m, \omega_m\right)\right\|_2\right)^2\langle\rho(0) \mid \rho(0)\rangle,
\end{equation}
where $\|\cdot\|_2$ is the $2$-norm of matrix which is defined as
\begin{equation}
    \|A\|_2=\max_{x\neq0}\frac{\|Ax\|_2}{\|x\|_2}=\sqrt{\lambda_{\max }\left(A^T A\right)}.
\end{equation}
If there exist constants $\Lambda,\tilde{\alpha}$ satisfying
\begin{equation}\label{eq: upper bound}
|\lambda(t, \omega)| \leqslant \Lambda \quad \text { and } \quad
\|\widehat{\widetilde{S}}(t, \omega)\| \leqslant \widetilde{\alpha} \quad \forall t \in \mathbb{R}^{+}, \quad \forall \omega \in \Omega ,
\end{equation}
then we have
\begin{equation}
    \mathbb{E}_{\Xi}\left\langle\rho_{\Xi}(t) \mid \rho_{\Xi}(t)\right\rangle\le \mathrm{e}^{\Lambda(1+\tilde{\alpha}^2)t}\langle\rho(0)|\rho(0)\rangle.
\end{equation}
Note that \cref{eq:intensity function}, we can choose
\begin{equation}\label{eq: Lambda}
    \Lambda=|\Omega|\max\left\{V/4,\gamma\right\},
\end{equation}
which is a $\mathcal{O}(N)$ number.

We need further analysis for the $\tilde{\alpha}$. First, we have
\begin{equation}
        \widetilde{S}(t, \omega)({\rho})=
        \begin{cases}
        
        -\mathrm{i}\sigma_{z}^{(j)}\sigma_{z}^{(k)}{\rho} & \text { if } \omega=(j, k,1) ,\\
        \mathrm{i}{\rho}\sigma_{z}^{(j)}\sigma_{z}^{(k)} & \text { if } \omega=(j, k,2) ,
        \end{cases}
\end{equation}
therefore, its straightening satisfies
\begin{equation}
    \widehat{\widetilde{S}}(t, \omega)=
    \begin{cases}
        -\mathrm{i}I\otimes\sigma_z^{(j)}\sigma_z^{(k)}  & \text { if } \omega=(j, k,1) ,\\
        \mathrm{i}{\sigma_z}^{(j)}{\sigma_z}^{(k)}\otimes I  & \text { if } \omega=(j, k,2) ,
    \end{cases}
\end{equation}
Then, from \cref{prop: straightening eigenvalue} and the above formula, we can choose 
\begin{equation}\label{eq: tilde_alpha}
    \tilde{\alpha}=\max\left\{\|\sigma_{z}\|_2^2,\|\sigma_{z}^{T}\|_2^2\right\}=1,
\end{equation}
which is a $\mathcal{O}(1)$ number.

By combining \cref{eq: variance,eq: inner of straightening rho,eq: upper bound,eq: Lambda,eq: tilde_alpha}, we can get an estimate of the variance
\begin{equation}\label{eq: variance estimate}
    \sigma^2\le\frac{1}{N_{\text{traj}}}\mathrm{e}^{\mathcal{O}(N)t}\langle\rho(0)|\rho(0)\rangle.
\end{equation}

\begin{rmk}
    Here we highlight the advantages of our approach. When the system is small, it is evident that ordinary ODE solvers, like the Runge-Kuta method, the Exponential method, and so on, are preferred. However, the Runge-Kuta and Exponential methods should store the total density matrix, so they need much storage space when the system is extensive. For example, assume that there are $N$ sites in the system, then the dimension of the Hilbert space is $2^N$. If we use the Runge-Kuta method, we need to store $2^N\times 2^N$ matrix; if the Exponential method is used, we need to store the straightening of a $2^N\times 2^N$ matrix, i.e., we need to store a $2^{2N}\times 2^{2N}$ matrix.

    However, we only need to store some small matrices in the QKMC method. Specifically, when $\rho_{\Xi}$ has the form of simple tensors(i.e. $\rho_{\Xi}\in\mathcal{M}_1$)
    \begin{equation}
        \rho_{\Xi}=\rho_1\otimes\rho_2\otimes\cdots\otimes\rho_N,
    \end{equation}
    then the action of $U(t,s)$ make it still in $\mathcal{M}_1$,
    and the action of $\widetilde{A}(t,\omega)$ also make it in $\mathcal{M}_1$.
    Hence, as long as $\rho_{\Xi}(0)\in\mathcal{M}_1$, then $\rho_{\Xi}(t)\in\mathcal{M}_1$ can be written as $\rho_{\Xi}(t)=\rho_{1,\Xi}(t)\otimes\rho_{2,\Xi}(t)\otimes\cdots\otimes\rho_{N,\Xi}(t)$. 
    
    The observer operator $\widehat{O}$ also belongs to $\mathcal{M}_1$:
    \begin{equation}
        \widehat{O}=\widehat{O}_1\otimes\widehat{O}_2\otimes\cdots\otimes\widehat{O}_N.
    \end{equation}
    Therefore, the expectation of $\widehat{O}$ can be evaluate as follows
    \begin{equation}
        \begin{aligned}
            \langle\widehat{O}\rangle&=\operatorname{Tr}\left(\widehat{O}\rho(t)\right)=\operatorname{Tr}\left(\widehat{O}\mathbb{E}_{\Xi}\rho_{\Xi}(t)\right)=\mathbb{E}_{\Xi}\operatorname{Tr}\left(\widehat{O}\rho_{\Xi}(t)\right)\\
            &=\mathbb{E}_{\Xi}\operatorname{Tr}\left(\widehat{O}_1\otimes\widehat{O}_2\otimes\cdots\otimes\widehat{O}_N\cdot\rho_{1,\Xi}(t)\otimes\rho_{2,\Xi}(t)\otimes\cdots\otimes\rho_{N,\Xi}(t)\right)\\
            &=\mathbb{E}_{\Xi}\left[\operatorname{Tr}\left(\widehat{O}_1\rho_{1,\Xi}(t)\right)\operatorname{Tr}\left(\widehat{O}_2\rho_{2,\Xi}(t)\right)\cdots\operatorname{Tr}\left(\widehat{O}_N\rho_{N,\Xi}(t)\right)\right].
        \end{aligned}
    \end{equation}
    As a result, we only need to store $N$ $2\times 2$ small matrices under every trajectory. In addition, this method is more suitable for parallel computation.

    In a word, the QKMC algorithm can calculate larger systems than conventional ODE solvers. 
\end{rmk}

\begin{rmk}
    From the \cref{eq: variance estimate}, we can see that the variance of the QKMC algorithm grows exponentially with time. This is the famous `` dynamics sign problem " that makes us can only calculate in a short time. Reducing this approach's variance is one of our future works.
\end{rmk}

\section{Numerical Results}\label{sec:numerical results}
\subsection{Examples with few sites}. We first consider some simple cases with only a few sites to verify the QKMC algorithm. The parameters are chosen as
\begin{equation}
    g=4,h=1,\gamma=1,V=0.01.
\end{equation}
Assuming that every site interacts only with its adjacent sites and all the interaction strengths are equal. Initially, we assume that sites have the same state $|\uparrow\rangle$,
\begin{equation}
    \rho(0)=
    \begin{pmatrix}
        1&0\\
        0&0
    \end{pmatrix}
    \otimes
    \begin{pmatrix}
        1&0\\
        0&0
    \end{pmatrix}
    \otimes\cdots\otimes
    \begin{pmatrix}
        1&0\\
        0&0
    \end{pmatrix},
\end{equation}
and we are concerned about the evolution of the population for the "all spin-down" state:
\begin{equation}
    p(t)=\operatorname{Tr}\left(\rho(t)B\right),\quad B=
    \begin{pmatrix}
        0&0\\
        0&1
    \end{pmatrix}
    \otimes
    \begin{pmatrix}
        0&0\\
        0&1
    \end{pmatrix}
    \otimes\cdots\otimes
    \begin{pmatrix}
        0&0\\
        0&1
    \end{pmatrix},
\end{equation}
For small $N$, the reference solutions are calculated by the $4$-order Runge-Kutta method with time step $\mathrm{d}t=0.01$.

In \cref{fig: few sites}, we show the numerical results using QKMC. The cases with two spins to five spins are considered, and $100000$ trajectories are used for all four cases. We can see that QKMC gets the same result as $4$-order Runge-Kutta.

\begin{figure}[htbp]
	\centering  
	\subfigbottomskip=2pt 
	\subfigcapskip=-5pt 
	\subfigure[$N=2$]{
		\includegraphics[width=0.48\linewidth]{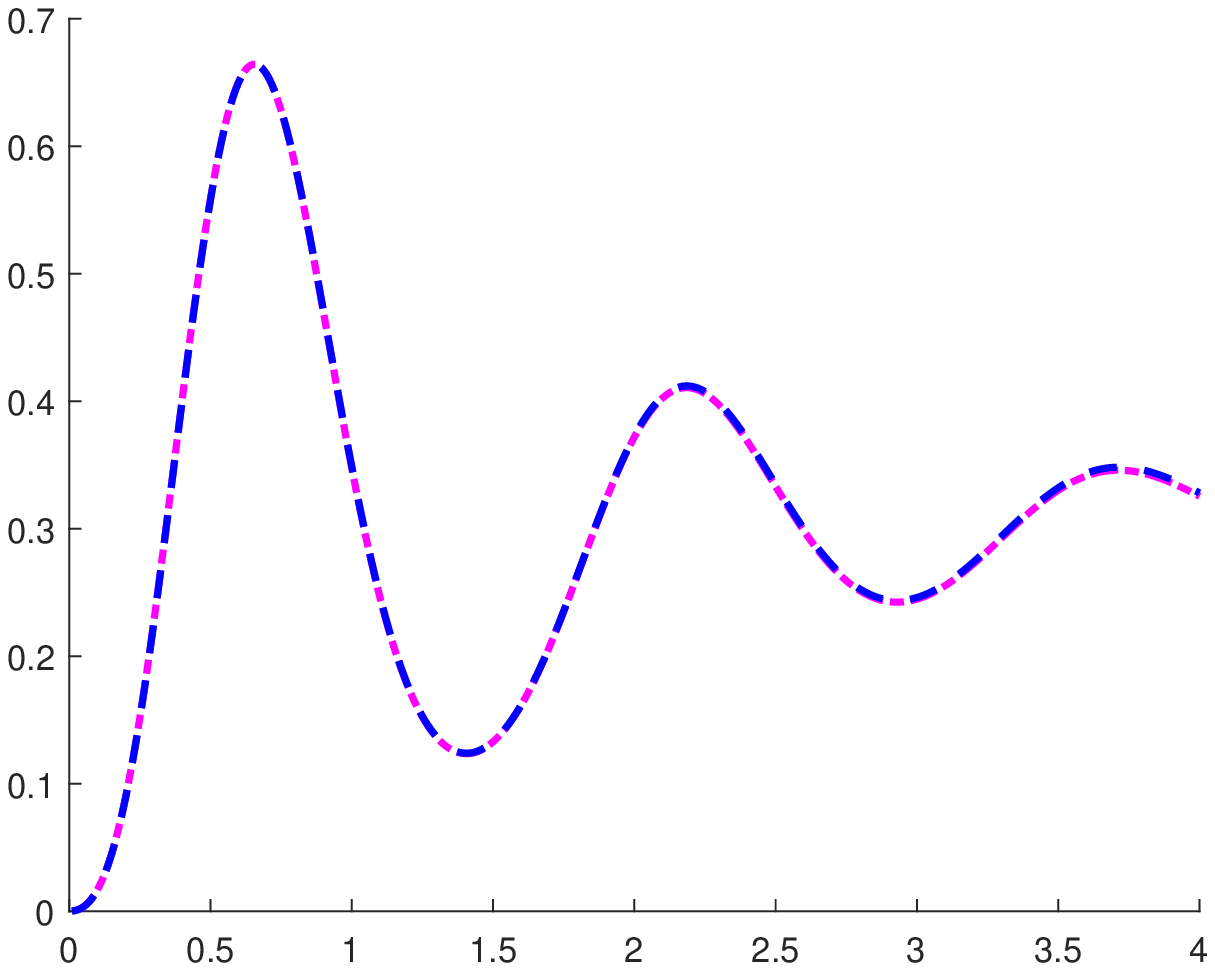}}
	\subfigure[$N=3$]{
		\includegraphics[width=0.48\linewidth]{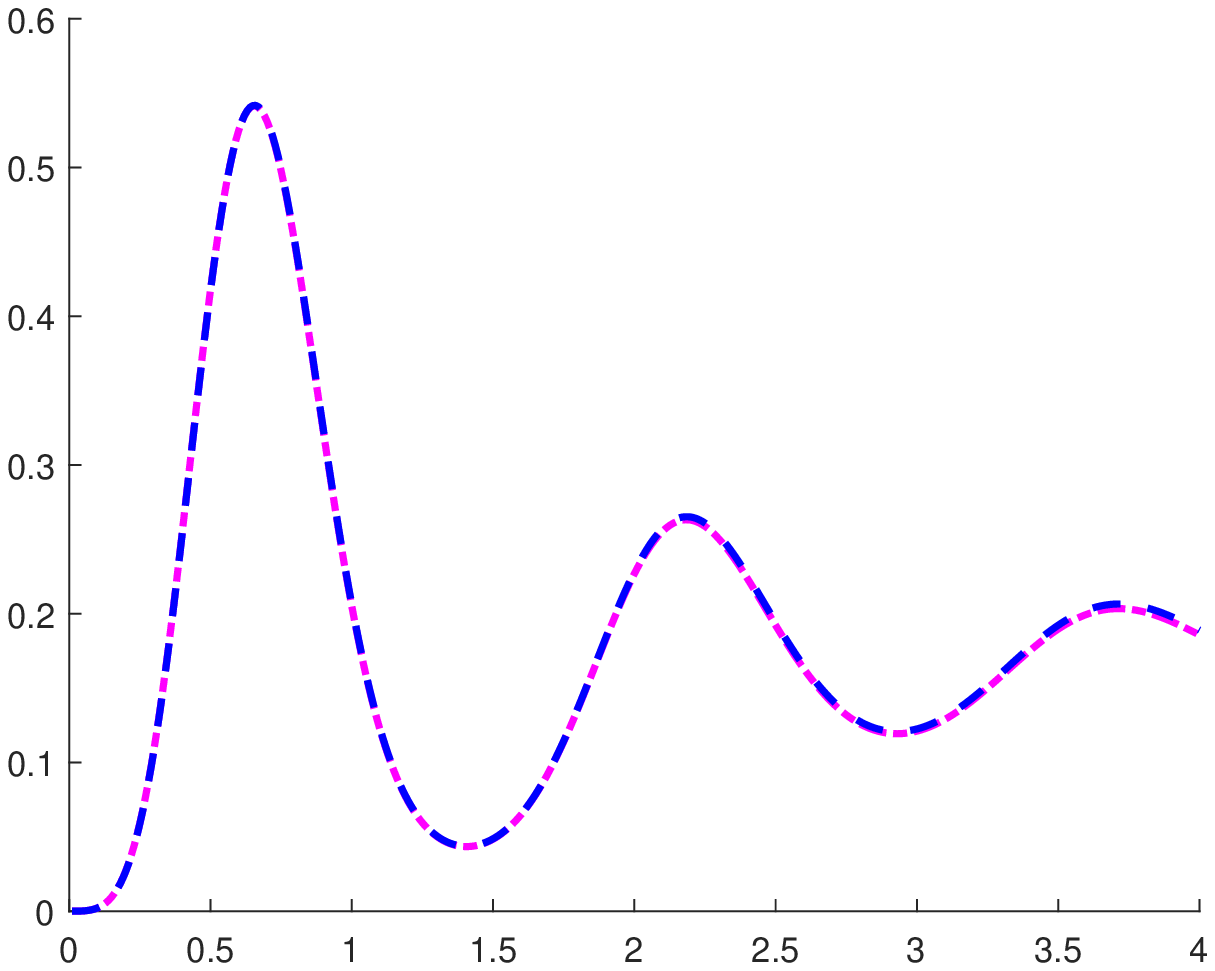}}
	  \\
	\subfigure[$N=4$]{
		\includegraphics[width=0.48\linewidth]{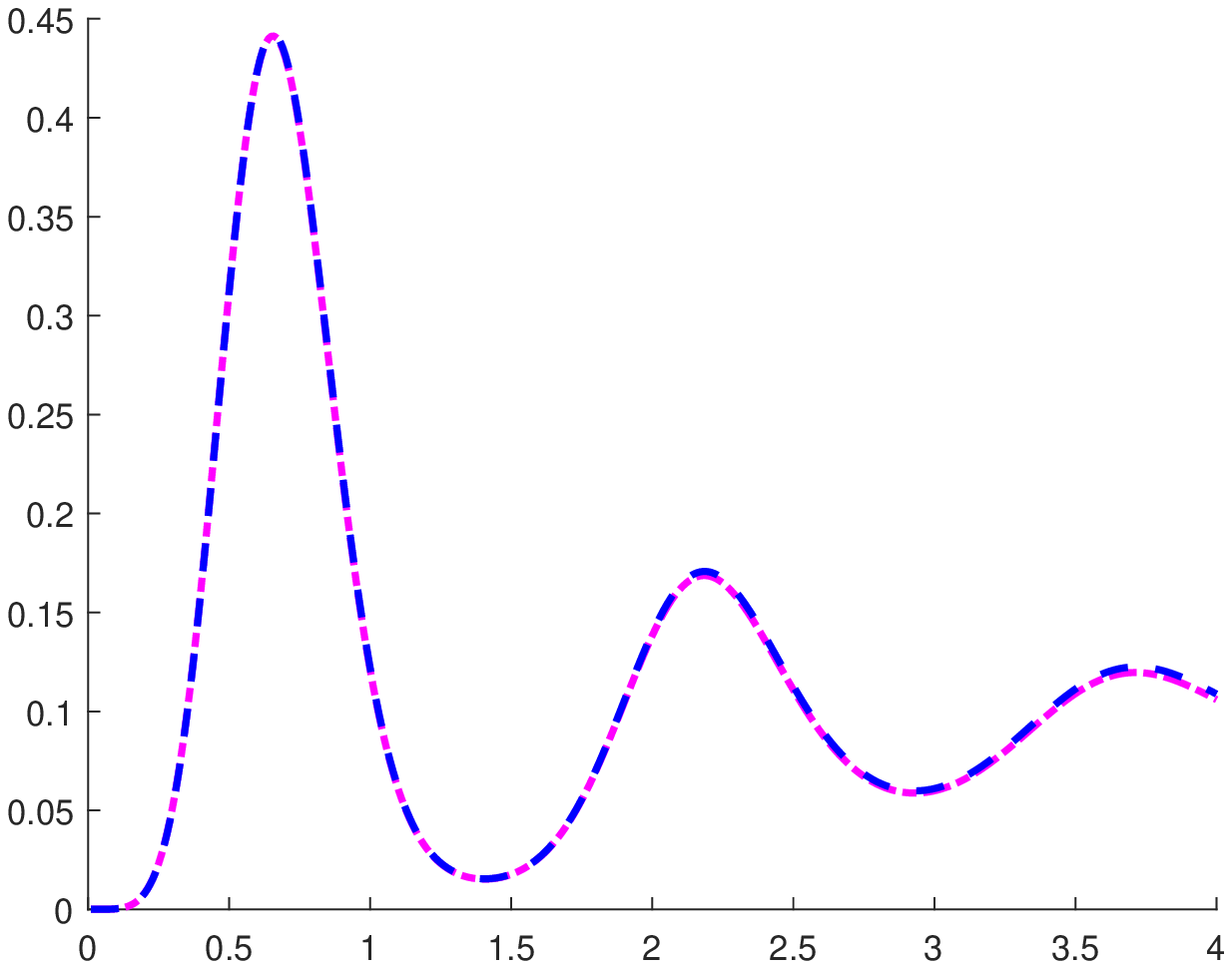}}
	\subfigure[$N=5$]{
		\includegraphics[width=0.48\linewidth]{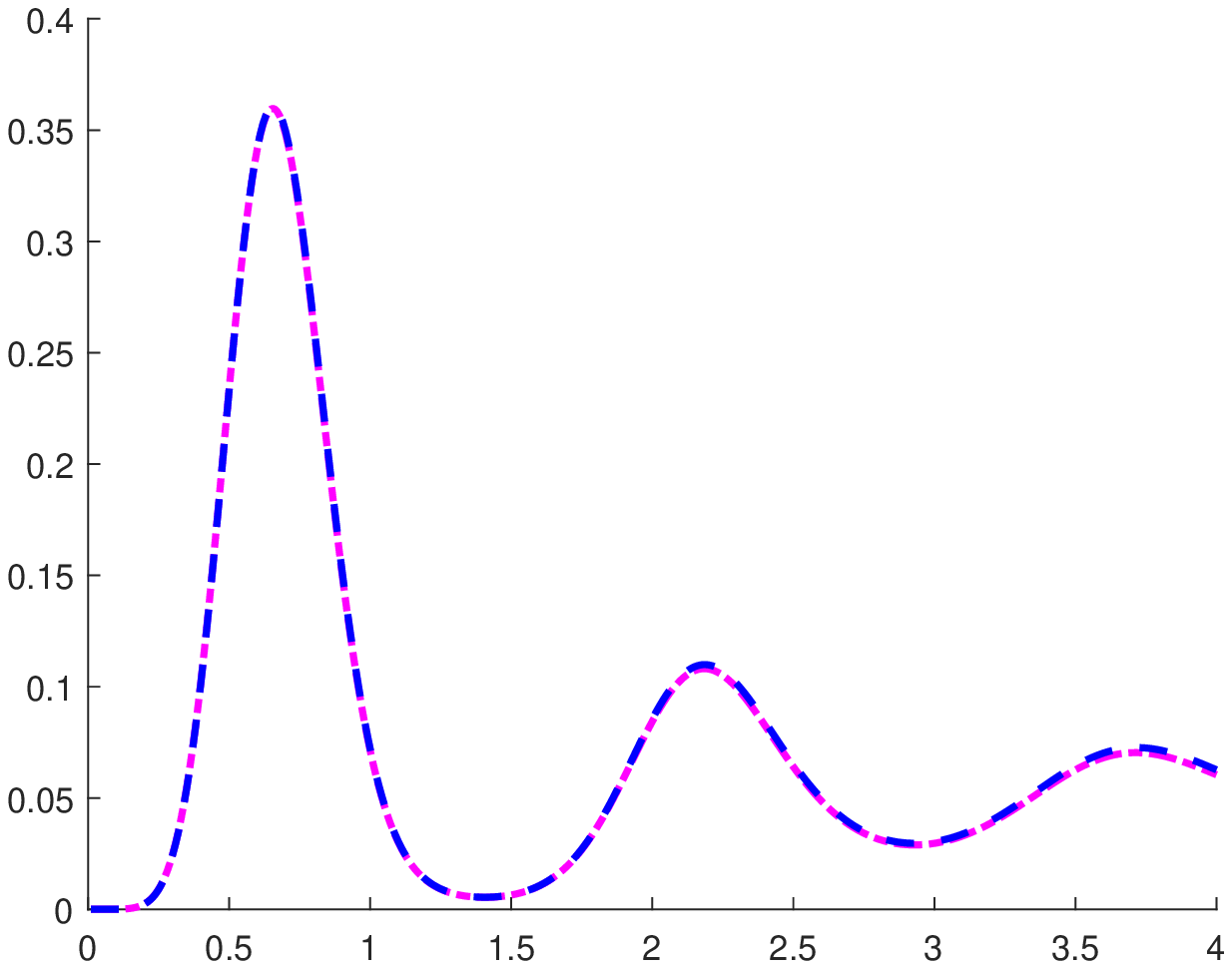}}
	\caption{The population of "all spin-down" state for few sites. The magenta line corresponds to the $4$-order Runge-Kutta method, and the blue line corresponds to the QKMC method.}
    \label{fig: few sites}
\end{figure}

In the four cases, we also check the numerical error for QKMC. Using the result of the $4$-order Runge-Kutta method as the reference solution $|\rho_{\mathrm{ref}}(t)\rangle$, we plot the evolution of the $2$-norm of numerical error $|\rho_{\mathrm{err}}(t)\rangle=|\rho_{\mathrm{num}}(t)\rangle-|\rho_{\mathrm{ref}}(t)\rangle$:
\begin{equation}
    e_{N_{\text{traj}}}(t)=\sqrt{\langle\rho_{\mathrm{err}}(t)|\rho_{\mathrm{err}}(t)\rangle},
\end{equation}
where $|\rho_{\mathrm{num}}(t)\rangle$ is the numerical solution of QKMC with $N_{\text{traj}}$ trajectories. \cref{fig: error} shows the error reduction as the number of trajectories increases, and the order of convergence is calculated in \cref{fig: order}. The numerical order is around $\frac{1}{2}$, which indicates that the expected convergence rate in the Monte Carlo method is achieved in our numerical test.

\begin{figure}[htbp]
    \centering
    \includegraphics[width=0.8\textwidth]{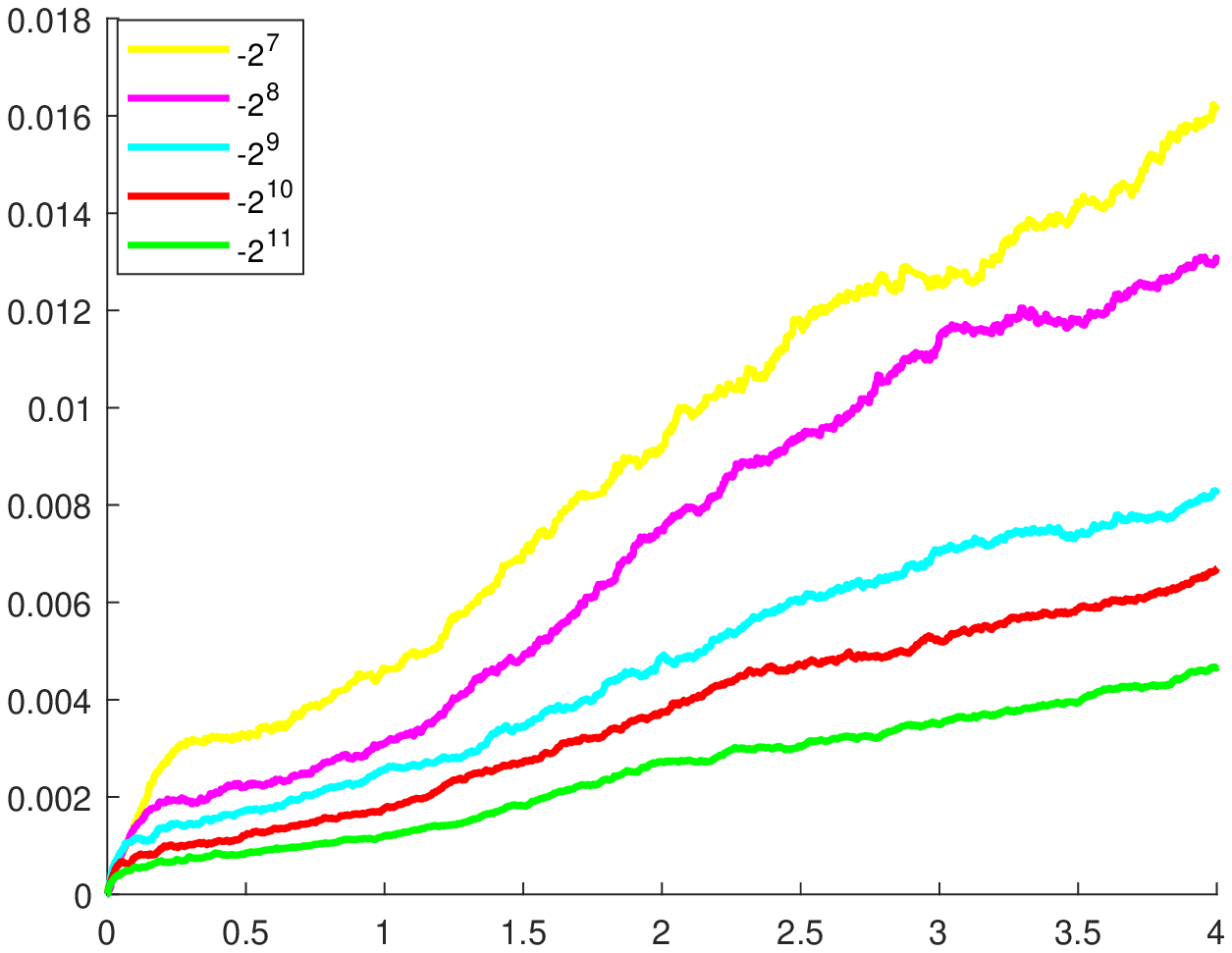}
    \caption{Evolution of the numerical error $e_{\text{Ntraj}}$ for two spins. The $x$-axis is the time $t$, and the $y$-axis is the error $e_{\text{Ntraj}}$.}
    \label{fig: error}
\end{figure}

\begin{figure}[htbp]
    \centering
    \includegraphics[width=0.8\textwidth]{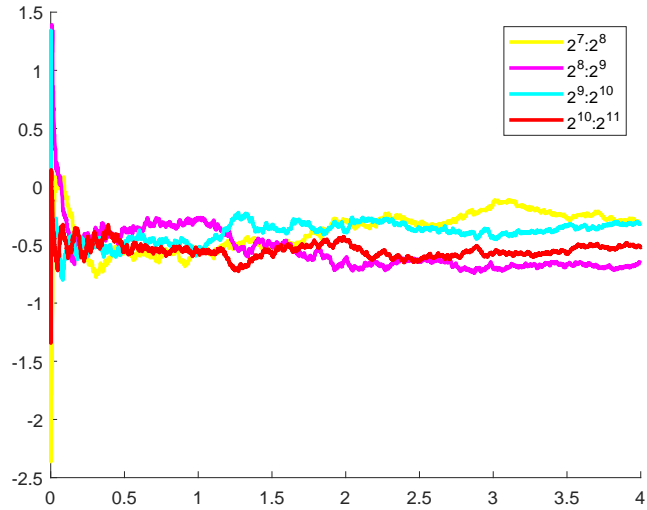}
    \caption{Evolution of the numerical order for two spins. The $x$-axis is the time $t$, and the $y$-axis is the order.}
    \label{fig: order}
\end{figure}

\subsection{Examples with more sites}
The cases with more sites are shown in \cref{fig: more sites}. Here we consider weaker coupling intensity
\begin{equation}
    g=4,h=1,\gamma=1,V=0.002,
\end{equation}
While the number of sites ranges from $10$ to $40$, we use $1000000$ trajectories in QKMC. For $10$ sites, the numerical results again agree with the reference results. For $20$ to $40$ sites, the reference solutions are not provided since the computational time and the storage space for the Runge-Kuta method are not affordable.

\begin{figure}[htbp]
	\centering  
	\subfigbottomskip=2pt 
	\subfigcapskip=-5pt 
	\subfigure[$N=10$]{
		\includegraphics[width=0.48\linewidth]{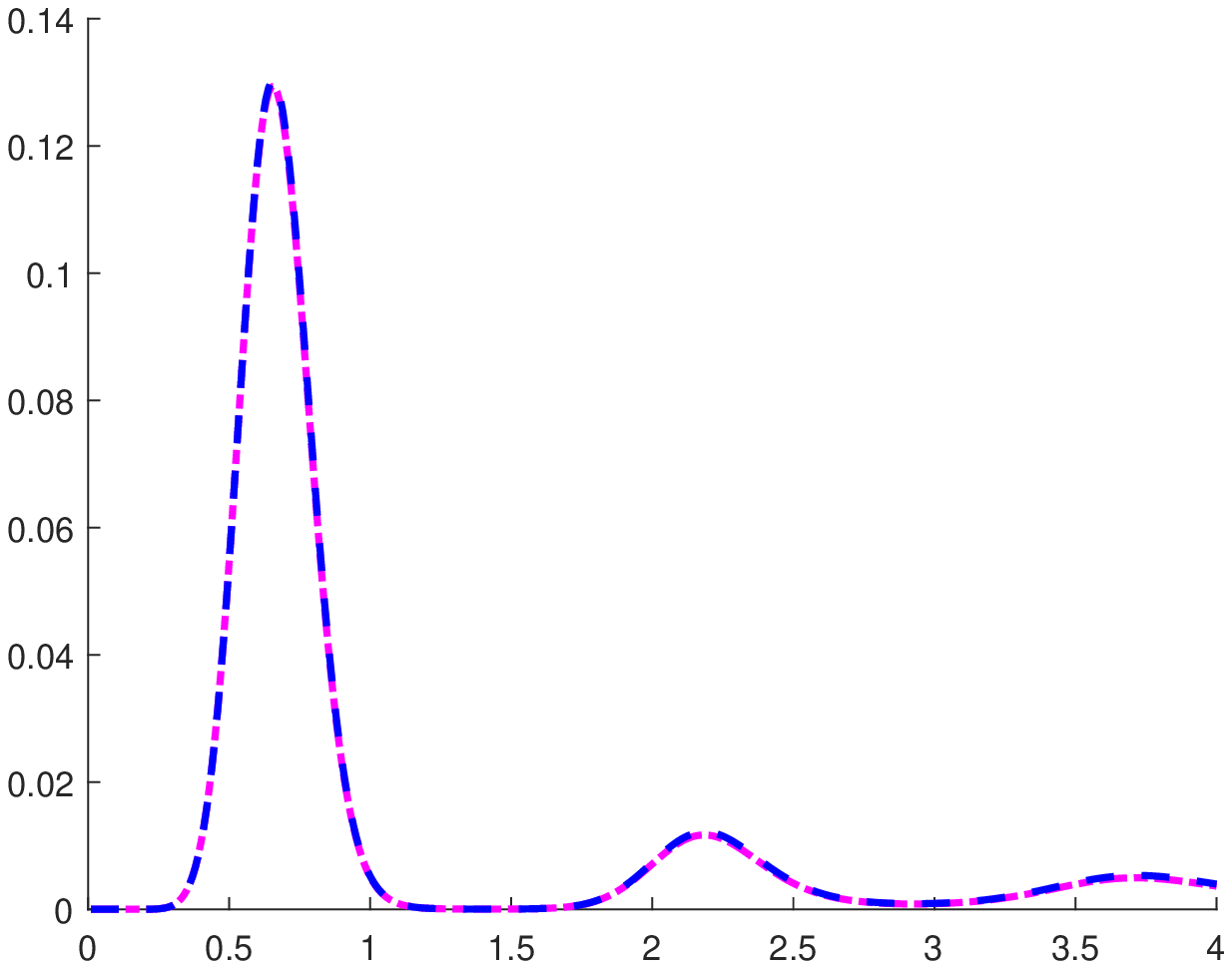}}
	\subfigure[$N=20$]{
		\includegraphics[width=0.48\linewidth]{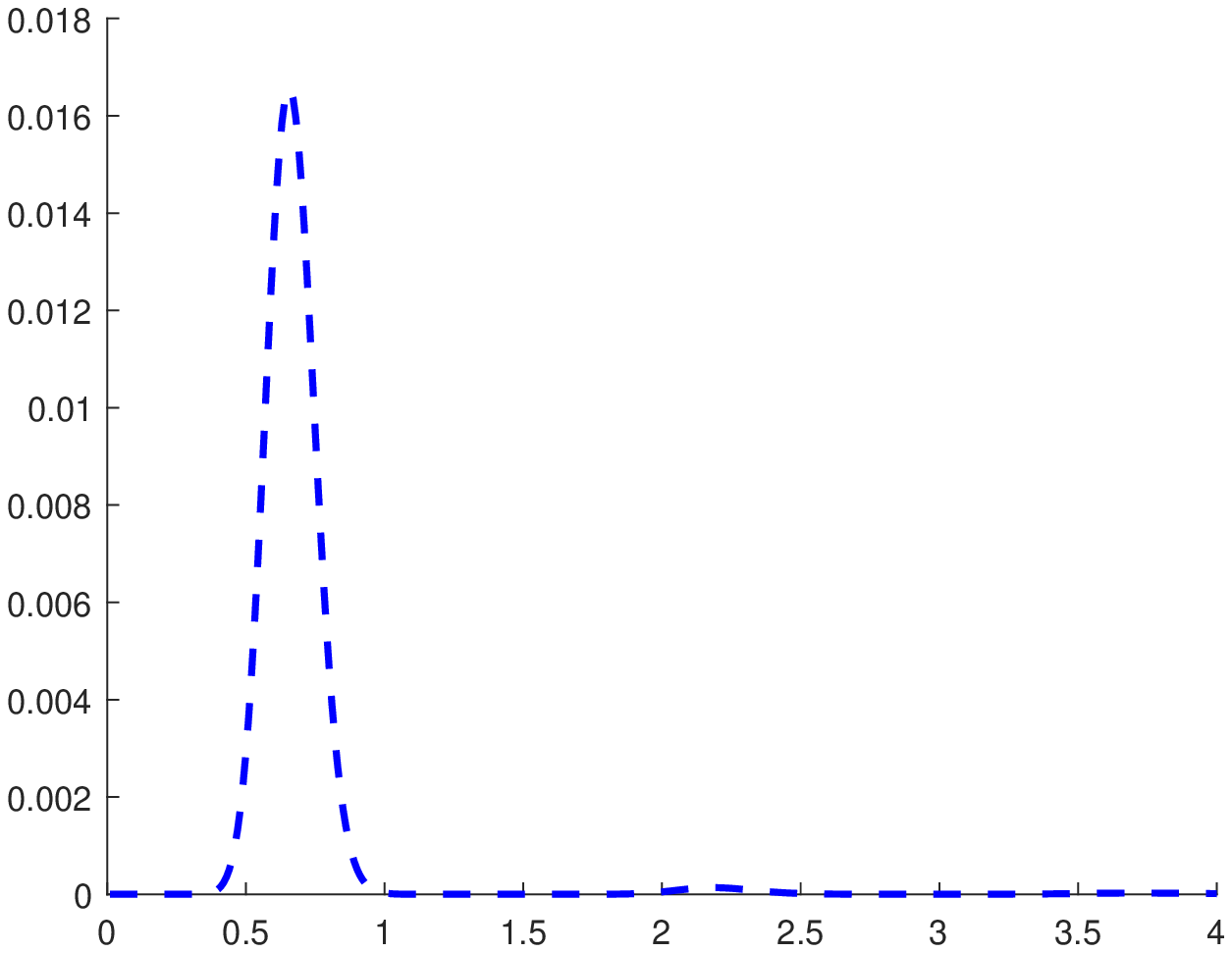}}
	  \\
	\subfigure[$N=30$]{
		\includegraphics[width=0.48\linewidth]{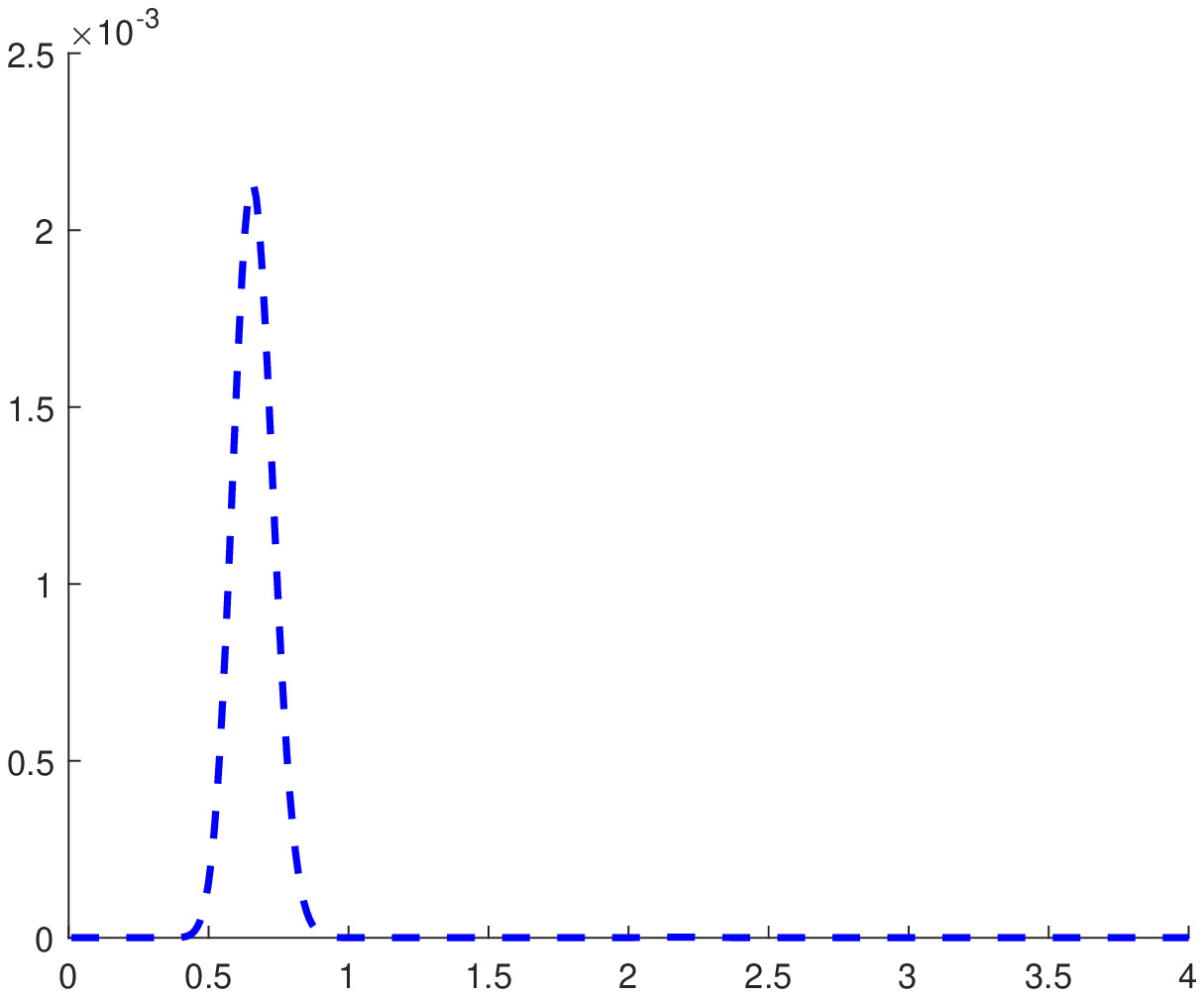}}
	\subfigure[$N=40$]{
		\includegraphics[width=0.48\linewidth]{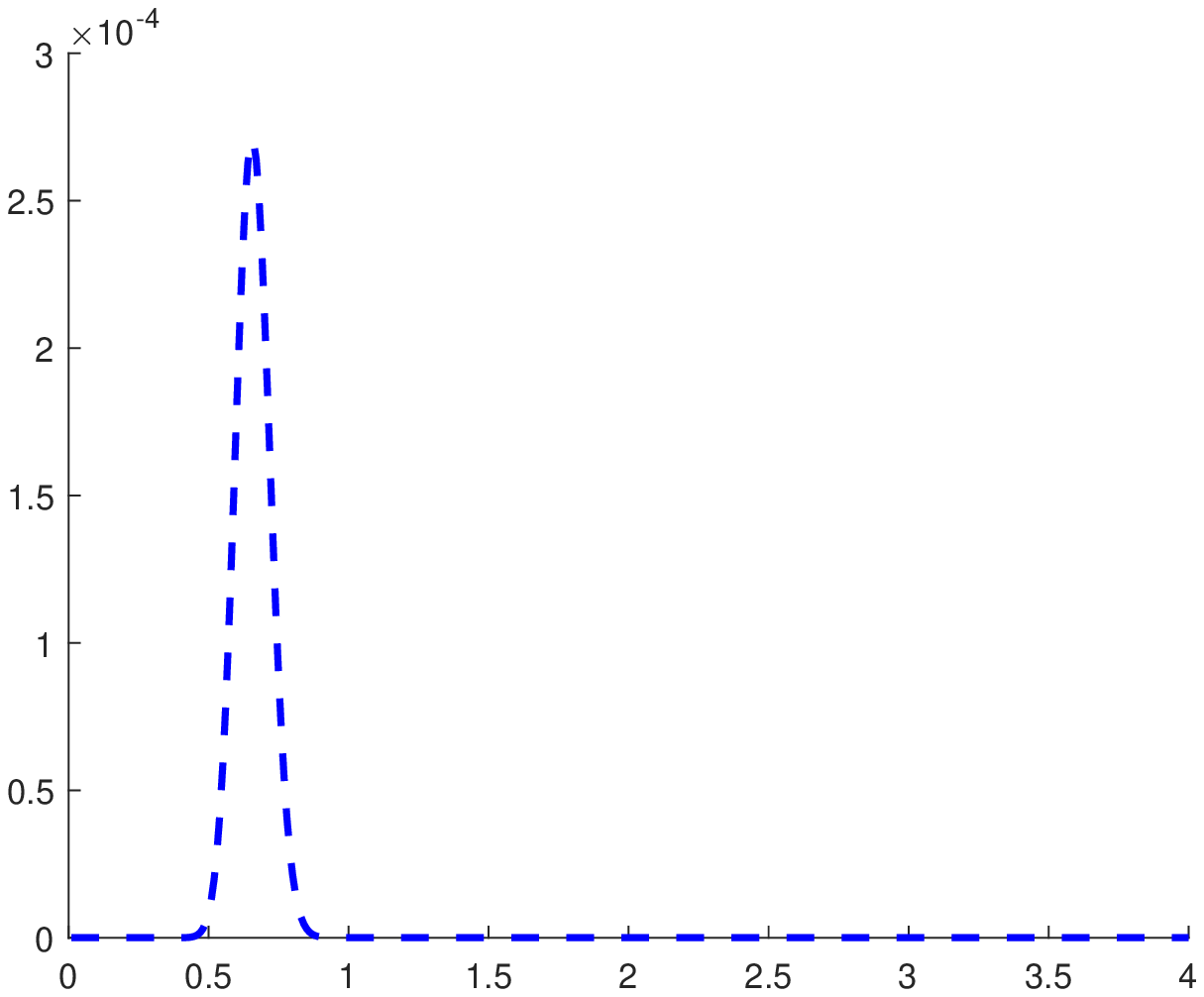}}
	\caption{The population of "all spin-down" state for more sites. The magenta line corresponds to the $4$-order Runge-Kutta method, and the blue line corresponds to the QKMC method.}
    \label{fig: more sites}
\end{figure}

\section{Conclusions}\label{sec:conclusion}
This article proposes a stochastic method to solve the Lindblad equation. This method makes full use of the tensor product structure of the matrices in the Lindblad equation, thus significantly reducing the storage cost. Numerical experiments show that the method can calculate larger systems than the deterministic solver and capture very small populations well. In the future, We will explore how to reduce the variance of this method and how to combine this method with the dynamical low-rank method and stochastic Schr$\mathrm{\ddot o}$dinger equation to reduce storage cost further.

\section{Acknowledgements}
The author thanks Dr. Zhennan Zhou for his valuable advice and fruitful discussions. Also, I am grateful to
Dr. Hao Wu for his support and encouragement.

\bibliographystyle{plain} 
\bibliography{QKMC_open} 

\end{document}